\documentclass[a4paper,USenglish]{lipics-v2018}
\usepackage{microtype}
\usepackage{multicol}
\usepackage{hyperref}
\hypersetup{
  colorlinks = true,
  linkcolor = blue,
  citecolor = red
}
\usepackage{xstring}
\usepackage[linesnumbered,vlined,algoruled]{algorithm2e}
\SetNoFillComment
\usepackage{amsmath}
\usepackage{amsthm}
\usepackage{graphicx}
\graphicspath{{./pics/}}
\graphicspath{{./pics/}}

\title{Reducing Compare-and-Swap to Consensus Number One Primitives}
\author{Pankaj Khanchandani}{ETH Zurich, Switzerland}{kpankaj@ethz.ch}{}{}
\author{Roger Wattenhofer}{ETH Zurich, Switzerland}{wattenhofer@ethz.ch}{}{}
\authorrunning{P. Khanchandani and R. Wattenhofer}
\Copyright{Pankaj Khanchandani and Roger Wattenhofer}
\subjclass{}
\nolinenumbers
\theoremstyle{plain} 
\theoremstyle{plain} 
\theoremstyle{plain} 
\theoremstyle{plain} 

\newcommand{\sref}[1]{\StrCut{#1}{:}\pre\post%
                        \IfStrEqCase{\pre}{%
                           {alg}{\hyperref[#1]{Algorithm~\ref*{#1}}}%
                           {lem}{\hyperref[#1]{Lemma~\ref*{#1}}}%
                           {thm}{\hyperref[#1]{Theorem~\ref*{#1}}}%
                           {ln}{\hyperref[#1]{Line~\ref*{#1}}}%
                           {tab}{\hyperref[#1]{Table~\ref*{#1}}}%
                           {cor}{\hyperref[#1]{Corollary~\ref*{#1}}}%
                           {fig}{\hyperref[#1]{Figure~\ref*{#1}}}%
                           {as}{\hyperref[#1]{Assumption~\ref*{#1}}}%
                           {sec}{\hyperref[#1]{Section~\ref*{#1}}}%
                           {def}{\hyperref[#1]{Definition~\ref*{#1}}}%
                           {lp}{\hyperref[#1]{Case~\ref*{#1}}}%
                         }[ss]%
                       }
\newcommand{\cas}[2]{\FuncSty{compare-and-swap(}\ArgSty{#1, #2}\FuncSty{)}}
\newcommand{\rd}{\FuncSty{read()}}
\newcommand{\wrt}[1]{\FuncSty{write(}\ArgSty{#1}{)}}
\newcommand{\hfmx}[1]{\FuncSty{half-max(}\ArgSty{#1}\FuncSty{)}}
\newcommand{\mxwr}[1]{\FuncSty{max-write(}\ArgSty{#1}\FuncSty{)}}
\newcommand{\mt}{\mathit}
\newcommand{\lay}{\underline{\phantom{v}}}

\newcommand{\p}{\,|\,}

\begin{document}

\maketitle

\begin{abstract}
  The consensus number of an object is the maximum number of processes among which binary consensus
  can be solved using any number of instances of the object and read-write registers. Herlihy
  \cite{herlihy:waitfreeSynchronization} showed in his seminal work that if an object has a
  consensus number of \(n\), then there is a universal construction for a \emph{wait-free} and
  \emph{linearizable} implementation of any non-trivial concurrent object or data structure that is
  shared among \(n\) processes. Thus, a synchronization object such as \emph{compare-and-swap} with
  an infinite consensus number and the corresponding  instruction can be viewed as ``strong''. On the
  other hand, a synchronization object such as \emph{fetch-and-add} with consensus number two and
  the corresponding fetch-and-add instruction can be viewed as ``weak''.

  Ellen et al.~\cite{ellen:complexityBasedHierarchy} observed recently that an object supporting \emph{two}
  weak instructions can also achieve infinite consensus number like an  object that supports \emph{one}
  strong instruction. Using Herlihy's universal construction, this implies that ignoring
  concerns about efficiency, one can design any concurrent data structure or algorithm using only
  weak instructions. However, is it possible that a combination of weak instructions is
  really powerful enough to \emph{efficiently} replace a strong instruction, like
  compare-and-swap, without incurring a large overhead in time or space? 

  In this paper, we answer this question by giving an \(O(1)\) time wait-free and linearizable
  implementation of a compare-and-swap register shared among \(n\) processes using read-write
  registers and registers that support two synchronization primitives \emph{half-max} and
  \emph{max-write}, each having consensus number one. The size of the registers required is
  logarithmic in the length of the execution. Thus, any algorithm that solves some arbitrary
  synchronization problem using read-write and compare-and-swap registers can be transformed into an
  algorithm that has the same asymptotic time complexity and uses registers that are logarithmic in
  the length of the execution and only support consensus number one instructions.
\end{abstract}

\section{Introduction}
Any multiprocessor chip needs to support some synchronization instructions, such as compare-and-swap
or fetch-and-add, to coordinate among several concurrent processes that can take steps
asynchronously at different rates. As it is not possible to support every other synchronization
instruction on a multiprocessor, the choice of instructions to support is important. Herlihy
\cite{herlihy:waitfreeSynchronization} gave an elegant way to make such a choice based on
\emph{consensus numbers}. The consensus number of an object is defined as the maximum number of
processes \(n\) among which \emph{binary consensus} can be solved using any number of instances of
the object and read-write registers. In binary consensus, each process is given an input of either
\(0\) and \(1\). Each process must output the same value (\emph{agreement}) within a finite number
of its steps (\emph{termination}) so that the output value is an input value of some process
(\emph{validity}).

Herlihy showed that objects of consensus number \(n\) can be used to construct a \emph{linearizable}
and \emph{wait-free} implementation of any concurrent data structure or object, such as stacks or
queues, shared among \(n\) processes. Linearizability implies that although each operation takes
several steps to complete, it appears to take effect instantaneously at some point between its
invocation and termination. The wait-free property implies that every process completes its
operation within a finite number of its steps irrespective of the speed of other processes.  As the
compare-and-swap object or register has infinite consensus number, supporting compare-and-swap on a
multiprocessor is a good and powerful choice. On the other hand, a fetch-and-increment object or
register has a consensus number of two and is an inherently weak choice by itself.

Recently, Ellen et al.~\cite{ellen:complexityBasedHierarchy} observed that the above classification
of synchronization instructions treats them as individual objects but in reality \emph{all} the
instructions supported by a multiprocessor can be applied on \emph{any} register or memory
location. They also give simple examples where two weak instructions can be combined on the same
object to achieve infinite consensus number. This along with Herlihy's universal construction
implies that it is possible to construct any concurrent data structure or object by only using weak
synchronization instructions. Although possible, such a construction would be inefficient both in
time and space. It is reasonable to argue that the weak instructions are only powerful enough to
solve consensus efficiently but not enough to \emph{efficiently} replace a strong instruction in Herlihy's
hieararchy.  In fact, in a followup work by Gelashvili et
al.~\cite{gelashvili:towardsReducedInstructionSets}, the authors write the following:

``\emph{The practical question is whether we can really replace a compare-and-swap instruction
in concurrent algorithms and data-structures with a combination of weaker instructions.}''

Note that when we refer to the consensus number of a synchronization instruction or a primitive, we
refer to the consensus number of an object that supports two operations: the synchronization
primitive and a read operation. It is essential that we also consider the read operation on the
object, otherwise, arbitrarily powerful primitives that do not return any value would have consensus
number one as there would be no way to read the object (for eg., a compare-and-swap primitive that
does not return a value). Thus, consensus number one primitives are like read-write registers where
the write operation is replaced with another weak ``write-like'' operation. The challenge is to come
up with similar weak operations that can be combined to  efficiently replace compare-and-swap.

In this paper, we show that it is possible to simulate a compare-and-swap register using a
combination of weak instructions and the simulation is efficient both in space and time. Concretely,
we introduce two consensus number one primitives \emph{half-max} and \emph{max-write}. We show that
using read-write registers and registers that support half-max and max-write, we can construct a
linearizable and wait-free implementation of a compare-and-swap register so that every
compare-and-swap operation takes \(O(1)\) time.
The size of the registers required is
logarithmic in the length of the execution. The total number of registers required is \(O(n)\) where
\(n\) is the number of processes. Thus, any \(O(T)\) algorithm using compare-and-swap and read-write
registers can be transformed into an \(O(T)\) time algorithm that only uses 
consensus number one instructions on reasonably large registers. We also outline an extension for
simulating \(m\) compare-and-swap registers where the total number of registers required is
\(O(m+n)\).


\section{Related Work}
One of the most central questions in concurrent computing has been to quantify the power of
synchronization instructions.  Herlihy \cite{herlihy:waitfreeSynchronization} originally defined the
consensus number of an object as the maximum number of processes \(n\) that can solve consensus
using a \emph{single} instance of the object and any number of read-write registers. As a
consequence of this definition, an object that has higher consensus number or is higher in the
Herlihy's hierarchy cannot be implemented using an object that has a lower consensus number or is
lower in the Herlihy's hierarchy. Jayanti \cite{jayanti:onRobustness} defined \emph{robustness} of a
hierarchy as the property that an object at a higher level in the hierarchy cannot be implemented
using \emph{any} number or combination of objects lower in the hierarchy. He gave an example of an
object such that \(k\) instances of the object along with read-write registers can solve consensus
for \(k + 1\) processes. Thus, Herlihy's hierarchy would not be robust if the consensus number
definition is restricted to use only single objects.

A natural fix is to allow any number of instances of the object in the definition of consensus
numbers, which is also the accepted definition and the one that we use
\cite{herlihy:artOfMultiprocessorProgramming}. Under this definition, Chandra et
al.~\cite{chandra:waitfreeVsTresilient} show that Herlihy's hierarchy is robust for two objects out of which
one of is a consensus object and the other one is an arbitrary object. Ruppert
\cite{ruppert:determiningConsensusNumbers} showed that Herlihy's hierarchy is robust for
read-modify-write and readable objects, which captures a large class of synchronization
primitives. All these results assume that when a set of objects are used to implement another
object, the synchronization operations supported by different objects are not merged onto a same
object.

Ellen et al.~\cite{ellen:complexityBasedHierarchy} observed that if one relaxes the above assumption
and does not treat a set of synchronization instructions as a set of individual objects but as a
single object supporting the set of synchronization instructions, then Herlihy's hierarchy is again
not robust. They propose a space based hierarchy in which the power of set of synchronization
instructions is quantified by the minimum amount of space required to solve obstruction free
consensus among \(n\) processes. A small value of this quantity for a set of synchronization
instructions means that the set of instructions is more powerful. This work has led to some more
followup work to understand the power of a set of synchronization instructions from different
perspectives when the instructions are assumed to be supported on the same
register.

In \cite{gelashvili:towardsReducedInstructionSets}, the authors give a lock-free implementation of a
log data structure by only using x86 instructions of consensus number at most two. They report that
the performance achieved was similar to that of a compare-and-swap based implementation. In our
work, we do not restrict ourselves to instructions supported on modern architecture
as our goal is to find if it is
theoretically possible to efficiently compete with a strong instruction
like compare-and-swap using low consensus number instructions only. In
\cite{khanchandani:lowConsensusSynchronization}, we observed that a set of low consensus number
instructions supported on the same register can help to improve the time bound of solving the
fundamental synchronization task of designing a wait-free queue from \(O(n)\) to \(O(\sqrt{n})\) for
\(n\) processes. 

In this paper, we look at the power of a set of low consensus number instructions supported on the
same register with respect to their ability to efficiently simulate a strong instruction like
compare-and-swap. We chose to simulate compare-and-swap not only because of its infinite consensus
number but also because it is ubiquitous and has been shown to yield efficient implementations
\cite{michael:waitfreeLLSC, jayanti:practicalWaitfreeLLSC, jayanti:logQueue}. Our result then
implies that a set of low consensus number instructions can be at least as powerful as
compare-and-swap. In \cite{golab:constantRMRImplementationOfCAS}, the authors give a blocking
implementation of comparison primitives by just using read-write registers and constant number of
remote memory references. Their focus is to use read-write registers and hence wait-freedom is
impossible to achieve. Overall, there is no prior work that shows that a set of low consensus number
instructions can be as powerful and efficient as compare-and-swap registers for an arbitrary
synchronization task.

\section{An Overview of the Method}
Our method is based on the observation that if several compare-and-swap
operations attempt to simultaneously change the value in the register, only one
of them succeeds. So, instead of updating the final value of the register for
each operation, we first determine the single operation that succeeds and update
the final value accordingly. This is achieved by using two consensus number one
primitives: \emph{max-write} and \emph{half-max}.

The max-write primitive
takes two arguments. If the first argument is greater than or equal to the value
in the first half of the register, then the first half of the register is replaced with the
first argument and the second half is replaced with the second argument.
Otherwise, the register is left unchanged. In any case, no value is returned. This primitive helps in
keeping a version number along with a value.

The half-max primitive takes a single argument and replaces the first half of the register with that
argument if the argument is larger. Otherwise, the register remains unchanged. Again, no value is
returned in any case. This primitive is used along with the max-write primitive to determine the
single successful compare-and-swap operation out of several concurrent ones.  The task of
determining the successful compare-and-swap operation can be viewed as a variation of tree-based
combining (as in \cite{ellen:optimalFetchAndIncrement, khanchandani:fastSharedCounting} for
example). The difference is that we do not use a tree as it would incur \(\Theta(\log n)\) time
overhead. Instead, our method does the combining in constant time as we will see later.

In the following section, we formalize the model and the problem. In
\sref{sec:alg}, we give an implementation of the compare-and-swap operation
using registers that support the half-max, max-write, read and write
operations. In \sref{sec:ana}, we prove its correctness and show that the
compare-and-swap operation runs in \(O(1)\) time. In \sref{sec:cn}, we argue that
the consensus numbers of the max-write and half-max primitives are both
one. Finally, we conclude and discuss some extensions in \sref{sec:conc}.

\section{Model}
A \emph{sequential object} is defined by the tuple \((S, O, R, T)\).  Here,
\(S\) is the set of all possible \emph{states} of the object, \(O\) is the set
of \emph{operations} that can be performed on the object, \(R\) is the set of
possible \emph{return values} of all the operations and
\(T: S \times O \to S \times R\) is the \emph{transition function} that
specifies the next state of the object and the return value given a state of the
object and an
operation applied on it.

A \emph{register} is a sequential object and supports the operations \emph{read}, \emph{write},
\emph{half-max} and \emph{max-write}. The read() operation returns
the current value (state) of the register. The write(\(v\)) operation updates
the value of the register to \(v\). The half-max(\(x\)) operation replaces the
value in the first half of the register, say \(a\), with \(\max\{x, a\}\) and does not return any
value. The
max-write(\(x\p y\)) operation replaces the first half of the register, say
\(a\), with \(x\) and second half of the register with \(y\) if and only if
\(x\geq a\). In any case, the operation does not return any value. The register operations are \emph{atomic}, i.e., if different
processes execute them simultaneously, then they execute sequentially in some
order. In general, atomicity is implied whenever we use the word operation in
the rest of the
text.

An \emph{implementation} of a sequential object is a collection of
\emph{functions}, one for each operation defined by the object.  A function
specifies a sequence of \emph{instructions} to be executed when the function is
executed.  An instruction is an operation on a register or a computation on
local variables, i.e., variables exclusive to a process.

A \emph{process} defines a sequence of instructions to be executed depending on the functions it
executes. The processes have identifiers \(1, 2, \ldots, n\).  When a process executes a function,
it is said to \emph{call} that function.  A \emph{schedule} is a sequence of process identifiers.
Given a schedule \(S\), an \emph{execution} \(E(S)\) is the sequence of instructions obtained by
replacing each process identifier in the schedule with the next instruction to be executed by the
corresponding process.

Given an execution and a function called by a process, the \emph{start} of the
function call is the point in the execution when the first register operation of
the function call appears. Similarly, the \emph{end} of the function call is the
point in the execution when the last register operation of the function call
appears.  A function call \(A\) is said to occur \emph{before} another function
call \(B\), if the call \(A\) ends before call \(B\) starts. Thus, the function
calls of an implementation of an object \(O\) form a partial order \(P_O(E)\)
with respect to an execution \(E\).  An implementation on an object \(O\) is
\emph{linearizable} if there is a total order \(T_O(E)\) that extends the
partial order \(P_O(E)\) for any given execution \(E\) so that the actual return
value of every function call in the order \(T_O(E)\) is same as the return
value determined by applying the specification of the object to the order
\(T_O(E)\). The total order \(T_O(E)\) is usually defined by associating a
\emph{linearization point} with each function call, which is a specific point in
the execution when the call takes effect. An implementation is \emph{wait-free}
if every function call returns within a finite number of steps of the calling
process irrespective of the schedule of the other processes.

Our goal is to develop a wait-free and linearizable implementation of the
\emph{compare-and-swap} register. It supports \emph{read} and
\emph{compare-and-swap} operations. The read() operation returns the current
value of the register. The compare-and-swap(\(a\), \(b\)) operation returns true
and updates the value of the register to \(b\) if the value in the register is
\(a\). Otherwise, it returns false and does not change the value.

\section{Algorithm}\label{sec:alg}

\sref{fig:sim} shows the (shared) registers that are used by the algorithm.  There
are arrays \(A\) and \(R\) of size \(n\) each. The \(i^{th}\) entry of the array
\(A\) consists of two \emph{fields}: the field \(c\) keeps a count of the number of
compare-and-swap operations executed by the process \(i\), the field
\(\mt{val}\) is used to store or \emph{announce} the second argument of the
compare-and-swap operation that the process \(i\) is executing. The \(i^{th}\)
entry of the array \(R\) consists of the fields \(c\) and \(\mt{ret}\). The
field \(\mt{ret}\) is used for storing the \emph{return} value of the \(c^{th}\)
compare-and-swap operation executed by the process \(i\). The register \(V\)
stores the current \emph{value} of the compare-and-swap object in the field
\(\mt{val}\) along with its version number in the field \(\mt{seq}\). The fields
\(\mt{seq}\), \(\mt{pid}\) and \(\mt{c}\) of the register \(P\) respectively
store the next version number, the \emph{process identifier} of the process that executed the
latest successful compare-and-swap operation and the count of compare-and-swap
operations issued by that process. For all the registers, the individual fields
are of equal sizes except for the register \(P\). The first half of this
register stores the field \(\mt{seq}\) where as the second half stores the other
two fields, \(\mt{pid}\) and \(c\).
 
\begin{figure}[!htb]
  \centering
  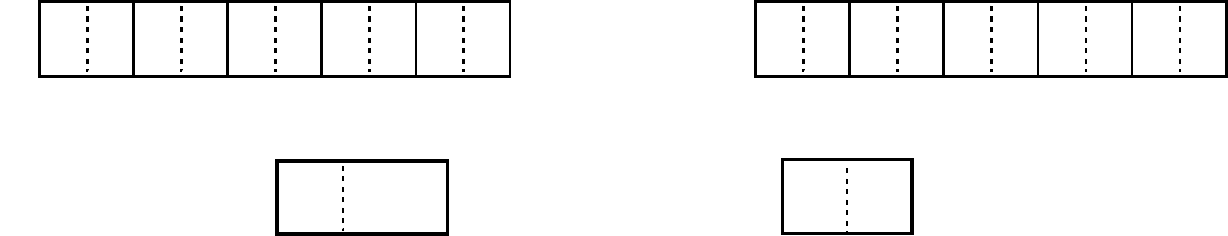
  \caption{An overview of data structures used by \protect\sref{alg:sim}.}
  \label{fig:sim}
\end{figure}

\sref{alg:sim} gives an implementation of the compare-and-swap register.  To
execute the read function, a process simply reads and returns the current value
of the object as stored in the register \(V\) (Lines \ref{ln:rd0} and
\ref{ln:rd1}). To execute the compare-and-swap function, a process starts by
reading the current value of the object (\sref{ln:rdval}). If the first argument
of the function is not equal to the current value, then it returns false (Lines
\ref{ln:neq} and \ref{ln:neqret}). If both the arguments are same as the current
value, then it can simply return true as the new value is same as the initial
one (Lines \ref{ln:eqeq} and \ref{ln:eqeqret}).

Otherwise, the process competes with the other processes executing the compare-and-swap function
concurrently. First, the process increments its local counter (\sref{ln:count}). Then, the new value
to be written by the process is announced in the respective entry of the array \(A\) (\sref{ln:ann})
and the return value of the function is initialized to false by writing to the respective entry in
the array \(R\) (\sref{ln:init}). The process starts competing with the other concurrent processes
by trying to announce its identifier in \(P\) using the max-write operation (\sref{ln:comp}). The
competition is finished by writing a version number larger than used by the competing processes
(\sref{ln:win}).

\begin{algorithm}[!htb]
  \SetKwProg{fn}{ }{ }{ }
  \SetKw{true}{true}
  \SetKw{false}{false}
  \SetKw{and}{and}

  \fn{\rd{}}{
    \((\,\lay{}\p \mt{val})\leftarrow V.\)\rd{}\;\label{ln:rd0}
    \KwRet{\(\mt{val}\)}\;\label{ln:rd1}
  }

  \fn{\cas{a}{b}}{\label{ln:arg}
    
    \((seq \p val) \leftarrow V.\)\rd{}\;\label{ln:rdval}
    \If{\(a \neq val\)}{\label{ln:neq}
      \KwRet{\false}\;\label{ln:neqret}
    }
    \If{\(a = b\)}{\label{ln:eqeq}
      \KwRet{\true}\;\label{ln:eqeqret}
    }
    \(c\leftarrow c + 1\)\;\label{ln:count}
    \(A[\mt{id}].\)\wrt{\(c\p b\)}\;\label{ln:ann}
    \(R[\mt{id}].\)\wrt{\(c \p \false\)}\;\label{ln:init}
    \(P.\)\mxwr{\(\mt{seq} +1 \p \mt{id} \p c\)}\;\label{ln:comp}
    \(P.\)\hfmx{\(\mt{seq} +2\)}\;\label{ln:win}
    \((\mt{seq} \p \mt{pid}\p \mt{cp}) \leftarrow P.\)\rd{}\;\label{ln:prd}
    \((\mt{ca}\p \mt{val}) \leftarrow A[\mt{pid}].\)\rd{}\;\label{ln:ard}
    \If{\(\mt{seq} \text{ is even } \and{}\text{ }\mt{cp} = \mt{ca}\)}{\label{ln:chk}
      \(R[\mt{pid}].\)\mxwr{\(\mt{ca} \p \true\)}\;\label{ln:inf}
      \(V.\)\mxwr{\(\mt{seq} \p \mt{val} \)}\;\label{ln:update}
    }
    \((\,\lay{}\p \mt{ret})\leftarrow R[\mt{id}].\)\rd{}\;\label{ln:ret}
    \KwRet{\(\mt{ret}\)}\;\label{ln:end}
  }
  \caption{The compare-and-swap and the read functions. The symbol \(\p\) is a
    field separator. The symbol \(\,\protect\lay{}\,\) is a variable that is
    not used. The variable \(\mt{id}\)
    is the identifier of the process executing the function. At initialization,
    we have \(c = 0\) and \(V = (0 \p x)\), where \(x\) is the initial value of the
    compare-and-swap object.}
  \label{alg:sim}
  \BlankLine
\end{algorithm}

Once the winner of the competing processes is determined, the winner and the
value announced by it is read (Lines~\ref{ln:prd} and \ref{ln:ard}), the winner
is informed that it won after appropriate checks (Lines~\ref{ln:inf}, \ref{ln:chk})
and the current value is updated (\sref{ln:update}). The value to be returned is
then read from the designated entry of array \(R\) (\sref{ln:ret}). A closer look at the algorithm reveals that the half-max and max-write operations
are only combined on the register \(P\). All other registers either only use max-write (and not
half-max) or are only read-write registers.

In the following section, we analyze \sref{alg:sim} and show that it is a
linearizable and \(O(1)\) time wait-free implementation of the compare-and-swap object.

\section{Analysis}\label{sec:ana}
Let us first define some notation. We refer to a field \(f\) of a register
\(X\) by \(X.f\). The term \(X.f_k^i\) is the value of the field \(X.f\) just
after process \(i\) executes Line~\(k\) during a call. We omit the call
identifier from the notation as it will be always clear from the context. Similarly, \(v_k^i\) is the
value of a variable \(v\), that is local to the process \(i\), just after it
executes Line~\(k\) during a call. The term \(X.f_e\) is the value of a field \(X.f\) at
the end of an execution.

\newcommand{\vseqr}[1]{V.\mt{seq}_{\ref*{ln:rdval}}^{#1}}
\newcommand{\vvalrdval}[1]{V.\mt{val}_{\ref*{ln:rdval}}^{#1}}
\newcommand{\vseqw}[1]{V.\mt{seq}_{\ref*{ln:update}}^{#1}}
\newcommand{\ppidr}[1]{\mt{pid}_{\ref*{ln:prd}}^{#1}}
\newcommand{\pseqr}[1]{\mt{seq}_{\ref*{ln:prd}}^{#1}}
\newcommand{\pseqw}[1]{\mt{seq}_{\ref*{ln:win}}^{#1}}
\newcommand{\vseqe}{V.\mt{seq}_e}
\newcommand{\seqchk}[1]{\mt{seq}_{\ref*{ln:chk}}^{#1}}
\newcommand{\vseqprd}[1]{V.\mt{seq}_{\ref*{ln:prd}}^{#1}}
\newcommand{\pidprd}[1]{\mt{pid}_{\ref*{ln:prd}}^{#1}}
\newcommand{\seqprd}[1]{\mt{seq}_{\ref*{ln:prd}}^{#1}}
\newcommand{\cprd}[1]{\mt{cp}_{\ref*{ln:prd}}^{#1}}
\newcommand{\card}[1]{\mt{ca}_{\ref*{ln:ard}}^{#1}}
\newcommand{\aarg}[1]{\mt{a}_{\ref*{ln:arg}}^{#1}}
\newcommand{\barg}[1]{\mt{b}_{\ref*{ln:arg}}^{#1}}

To prove that our implementation is linearizable, we first need to define the
linearization points. The linearization point of the compare-and-swap function
executed by a process \(i\) is given by \sref{def:linp}.  There are four main
cases.  If the process returns from \sref{ln:neqret} or \sref{ln:eqeqret}, then
the linearization point is the read operation in \sref{ln:rdval} as such an
operation does not change the value of the object (Cases~\ref{lp:1} and
\ref{lp:2}). Otherwise, we look for the execution of \sref{ln:update} that
wrote the sequence number \(\vseqr{i} + 2\) to the field \(V.\mt{seq}\) for the first
time. This is the linearization point of the process \(i\) if its
compare-and-swap operation was successful as determined by the value of
\(P.pid\) (\sref{lp:3a}). Otherwise, the failed compare-and-swap operations are
linearized just after the successful one (\sref{lp:3b}). The calls that have not
 taken effect are linearized after all other linearization points (\sref{lp:4}).

\begin{definition}\label{def:linp}
  The compare-and-swap call by a process \(i\) is
  linearized as follows.
\begin{enumerate}
\item If \(\vvalrdval{i} \neq a_{\ref*{ln:arg}}^i\), then the linearization
  point is the point
  when \(i\) executes \sref{ln:rdval}.\label{lp:1}
\item If \(\vvalrdval{i} = a_{\ref*{ln:arg}}^i = b_{\ref*{ln:arg}}^i\), then 
  the linearization point is the point when \(i\) executes \sref{ln:rdval}.\label{lp:2}
\item If \(\vvalrdval{i} = a_{\ref*{ln:arg}}^i \neq b_{\ref*{ln:arg}}^i\) and
  \(\vseqe \geq \vseqr{i} + 2 \), then let \(p\) be the point when
  \sref{ln:update} is executed by a process \(j\) so that
  \(\vseqw{j} = \vseqr{i} + 2\) for the first time.\label{lp:3}
  \begin{alphaenumerate}
  \item If \(\ppidr{j} = i\), then the linearization point is \(p\).\label{lp:3a}
  \item If \(\ppidr{j} \neq i\), then the linearization point is just after\label{lp:3b}
    \(p\).
  \end{alphaenumerate}
\item If \(\vvalrdval{i} = a_{\ref*{ln:arg}}^i \neq b_{\ref*{ln:arg}}^i\) and \(\vseqe < \vseqr{i} + 2 \), then the linearization point is at the
    end, after all the other linearization points in some order.\label{lp:4}
\end{enumerate}
\end{definition}

Note that we assume in \sref{lp:3} that if \(\vseqe \geq \vseqr{i} + 2\), then
there is an execution of \sref{ln:update} by a process \(j\) with the value
\(\vseqw{j} = \vseqr{i} + 2\). So, we first show in the following lemmas that
this is indeed true.

\begin{lemma}\label{lem:veven}
  The value of \(V.\mt{seq}\) is always even.
\end{lemma}
\begin{proof}
  We have \(V.\mt{seq} = 0\) at initialization. The modification only happens in
  \sref{ln:update} with an even value.
\end{proof}

\begin{lemma}\label{lem:vinc2}
  Whenever \(V.\mt{seq}\) changes, it increases by \(2\).
\end{lemma}
\begin{proof}
  Say that the value of the field was changed to \(\vseqw{i}\) when a process
  \(i\) executed \sref{ln:update}. Then, the value \(\seqchk{i}\) is even and so
  is the value \(\vseqw{i}\). Thus, the value \(\vseqw{i}\) was written to
  \(P.\mt{seq}\) by a process \(j\) and that \(\vseqr{j} = \vseqw{i} - 2\). As
  the field \(V.\mt{seq}\) is only modified by a max-write operation so it only
  increases. Thus, we have \(V.\mt{seq} \geq \vseqw{i} - 2\) just before \(i\)
  modifies it. As \(V.\mt{seq}\) is even by \sref{lem:veven} and \(i\) modifies
  it, we have \(V.\mt{seq} = \vseqw{i} - 2\) before the modification. So, the
  value increases by \(2\).
\end{proof}

\begin{lemma}
  The linearization point as given by \sref{def:linp} is well-defined.
\end{lemma}
\begin{proof}
  The linearization point as given by \sref{def:linp} clearly exists for
  all the cases except for \sref{lp:3}. For \sref{lp:3}, we only need to show that if
  \(\vseqe \geq \vseqr{i} + 2\), then there exists an execution of
  \sref{ln:update} by a process \(j\) so that \(\vseqw{j} = \vseqr{i} + 2\). As
  \(\vseqr{i}\) is even by \sref{lem:veven} and the value of \(V.\mt{seq}\) only
  increases in steps of \(2\) by \sref{lem:vinc2}, it follows from
  \(\vseqe \geq \vseqr{i} + 2\) that \(\vseqr{i} + 2\) was written to
  \(V.\mt{seq}\) at some point.
\end{proof}

To show that the implementation is linearizable, we need to prove two main
statements. First, the linearization point is within the start and end of the
corresponding function call. Second, the value returned by a finished call is same as
defined by the sequence of linearization points up to the linearization point of
the call. In the following two lemmas, we show the first of these statements.

\begin{lemma}\label{lem:vendcall}
  If the condition
  \(\vvalrdval{i} = a_{\ref*{ln:arg}}^i \neq b_{\ref*{ln:arg}}^i\) is true for a
  compare-and-swap call by a process \(i\), then the value of \(V.\mt{seq}\) is
  at least \(\vseqr{i} + 2\) at the end of the call.
\end{lemma}
\begin{proof}
  We define a set of processes \(S = \{j: \vseqr{j} = \vseqr{i}\}\). Consider
  the process \(k \in S\) that is the first one to execute \sref{ln:chk}.
  As the first field of \(P.\mt{seq}\) is always modified by a max operation and
  process \(k\) writes \(\vseqr{i} + 2\) to that field, we have \(\seqchk{k} = \pseqr{k} \geq
  \vseqr{i} + 2\). If  \(\pseqr{k} >  \vseqr{i} + 2\),
  then \(\vseqprd{k} \geq \vseqr{i} + 2\) and we are done.

  So, we only need to check the case when \(\pseqr{k} = \vseqr{i} + 2\). As
  \(\vseqr{i}\) is even by \sref{lem:veven}, so is \(\seqchk{k} = \pseqr{k}\).
  Moreover, the process \(\pidprd{k} \in S\) as some process(es) (including
  \(k\)) executed \sref{ln:comp}. As \(A[\pidprd{k}].c\) always increases
  whenever modified (\sref{ln:ann}), we have \(\card{k} \geq \cprd{k}\). But, if
  \(\card{k} > \cprd{k}\), then the process \(\pidprd{k}\) finished even before
  the process \(k\), a contradiction. So, it holds that \(\card{k} = \cprd{k}\)
  and the process \(k\) executes \sref{ln:update}.

  Now, the execution of \sref{ln:update} by the process \(k\) either changes the
  value of \(V.\mt{seq}\) or does not. If it does, then \(\vseqw{k}
  = \vseqr{i} + 2\) and we are done. Otherwise, someone already changed the
  value of \(V.seq\) to at least \(\vseqr{i} + 2\) because of \sref{lem:vinc2}.
\end{proof}

\begin{lemma}\label{lem:dur}
  The linearization point as given by \sref{def:linp} is within the corresponding
  call duration.
\end{lemma}
\begin{proof}
  The statement is true for Cases~\ref{lp:1} and \ref{lp:2} as the instruction
  corresponding to the linearization point is executed by the process \(i\)
  itself.
  
  For \sref{lp:3}, we analyze the case of finished and unfinished call
  separately.  Say that the call is unfinished. As
  \(V.\mt{seq}_e \geq \vseqr{i} + 2\) and \(\vseqr{i}\) is the value of
  \(V.\mt{seq}\) at the start of the call, the linearization point as given by
  \sref{def:linp} is after the call starts.  Now, assume that the call is
  finished. We know from \sref{lem:vendcall} that the value of \(V.\mt{seq}\) is
  at least \(\vseqr{i} + 2\) when the call ends. So, the point when
  \sref{ln:update} writes \(\vseqr{i} + 2\) to \(V.\mt{seq}\) is within the call
  duration.

  We know from \sref{lem:vendcall} that if the call finishes, then we have
  \(V.\mt{seq}_e \geq \vseqr{i} + 2\). So, if \(V.\mt{seq}_e < \vseqr{i} + 2\),
  then the call is unfinished and it is fine to linearize it at the end as done
  for \sref{lp:4}.
\end{proof}

Now, we need to show that the value returned by the calls is same as the value
determined by the order of linearization points. We show this in the following
lemmas.

\newcommand{\lpk}{\(\mt{LP}_k\)}
\newcommand{\lpkd}{\(\mt{LP}_{k'}\)}
\newcommand{\vseqk}[1]{V.\mt{seq}_{#1}}
\newcommand{\vvalk}[1]{V.\mt{val}_{#1}}
\newcommand{\vseqkd}{\(V.\mt{seq}_{k'}\)}
\newcommand{\vvalkd}{\(V.\mt{val}_{k'}\)}
\newcommand{\vseq}{\(V.\mt{seq}\)}
\newcommand{\vval}{\(V.\mt{val}\)}
\newcommand{\valard}[1]{\mt{val}_{\ref*{ln:ard}}^{#1}}

\begin{lemma}\label{lem:ppideqifseqeq}
  Assume that \(x = \seqprd{i} = \seqprd{j}\) for two distinct processes \(i\)
  and \(j\) and that \(x\) is even. Then, it implies that \(\pidprd{i} =
  \pidprd{j}\) and \(\cprd{i} = \cprd{j}\).
\end{lemma}
\begin{proof}
  Without loss of generality assume that the process \(i\) executes \sref{ln:prd} before the
  process \(j\) does so. As \(x = \seqprd{i} = \seqprd{j}\) by assumption, the
  only way in which the field \(P.\mt{pid}\) can change until the process \(j\)
  executes \sref{ln:prd}, is by a max-write operation on \(P\) with the value
  \(x\) as the first field. This is not possible as \(x\) is even and the
  max-write on \(P\) is only executed with odd value as the first field
  (\sref{ln:comp}). So, it holds that \(\pidprd{i} = \pidprd{j}\). Similarly, we
  have \(\cprd{i} = \cprd{j}\).
\end{proof}

\begin{lemma}\label{lem:vseqeqthenvvaleq}
  As long as the value of \vseq{} remains same, the value of \vval{} does not change.
\end{lemma}
\begin{proof}
  Say that a process \(i\) is the first one to write a value \(x\) to
  \vseq{}. The value written to the field \vval{} by the process \(i\) is
  \(\valard{i}\).  To have a different value of \vval{} with \(x\)
  as the value of \vseq{}, another process \(j\) must execute \sref{ln:update} with
  \(\seqprd{j} = x\) but
  \(\valard{i} \neq \valard{j}\). As \(\seqprd{j} = x = \seqprd{i}\),
  it follows from \sref{lem:ppideqifseqeq} that \(\pidprd{j} = \pidprd{i}\) and
  \(\cprd{i} = \cprd{j}\).
  As the condition in \sref{ln:chk} is true for  both the processes \(i\) and
  \(j\), it then follows that \(\card{i} = \card{j}\).
  As the field \(A[\pidprd{j}].\mt{val}\) is updated only once for a given value
  of \(A[\pidprd{j}].c\) (\sref{ln:ann}), it holds
  that \(\valard{i} = \valard{j}\) and the claim follows.
\end{proof}

\begin{lemma}\label{lem:rlookup}
  Say that \(\seqprd{i} = x\) is even and \(\pidprd{i} = j\) during a call by a process
  \(i\), then \(\vseqr{j} = x - 2\) for some call by process \(j\).
\end{lemma}
\begin{proof}
  As \(\seqprd{i} = x\), some process \(h\) modified \(P\) by executing
  \sref{ln:comp} or \sref{ln:win} with \(x\) as the first argument.
  As \(x\) is even and \(\vseqr{h}\) is even by \sref{lem:veven}, the process
  \(h\) modified \(P\) by executing \sref{ln:win}.
  So, it holds that \(\vseqr{h} = x - 2\).
  Also, process \(h\) executed \sref{ln:comp} with \(x-1\) as the first
  field.
  As \(\pidprd{i} = j\), the process \(j\) also executed \sref{ln:comp} with \(x -
  1\) as the first field after the process \(h\) did so.
  So, it holds that \(\vseqr{j} = x - 2\).
\end{proof}

\begin{lemma}\label{lem:vevenisalinp}
  For every even value \(x \in [2, V.\mt{seq}_e]\), there is an execution of
  \sref{ln:update} by a process \(i\) so that \(\seqprd{i} = x\) and the first
  such execution is the linearization point of some call.
\end{lemma}
\begin{proof}
  Consider an even value \(x \in [2, \vseqe{}]\). Then, we know from
  \sref{lem:vinc2} that \(x\) is written to \vseq{} by an execution of
  \sref{ln:update}. Let \(p\) be the point of first execution of
  \sref{ln:update} by a process \(j\) so that \(\seqprd{j} = x\). So, it holds
  for the process \(\pidprd{j} = h\) that \(\vseqr{h} = x - 2\) using
  \sref{lem:rlookup}.
  As point \(p\) is the first time when \(x\) is written to the field
  \(V.\mt{seq}\), it holds that \(\vseqw{j} = x\).
  Thus, \(p\) is the linearization point of the process
  \(h\) by \sref{def:linp}.
\end{proof}

\begin{lemma}\label{lem:vvalmodatlp}
  The value \(V.\mt{val}\) is only modified at a \sref{lp:3a} linearization
  point.
\end{lemma}
\begin{proof}
  Let \(q\) be a \sref{lp:3a} linearization point. Say that the value of
  \(V.\mt{seq}\) is updated to \(x\) at \(q\).
  Let \(p\) be the first point in the execution when the value of \(V.\mt{seq}\)
  is \(x - 2\).
  Using \sref{lem:vevenisalinp}, we conclude
  that \(p\) is either a 
  linearization point (for \(x - 2 \geq 2\)) or the initialization point (for \(x
  - 2 = 0\)). 
  Using \sref{lem:vseqeqthenvvaleq}, the value of \(V.\mt{val}\) is not modified
  between \(p\) and \(q\).
\end{proof}

\newcommand{\lval}[1]{L.\mt{val}_{#1}}
\newcommand{\lret}[1]{L.\mt{ret}_{#1}}

We want to use the above lemma in an induction argument on the linearization
points to show that the values returned by the corresponding calls are
correct. First, we introduce some notation for \(k \geq 1\). The term
\(\lval{k}\) is the value of the abstract compare-and-swap object after the
\(k^{th}\) linearization point. The terms \(\vseqk{k}\) and \(\vvalk{k}\),
respectively, are the values of \vseq{} and \vval{} after the \(k^{th}\)
linearization point.  These terms refer to the respective values just
after initialization for \(k = 0\). For \(k \geq 1\), the term \(\lret{k}\) is
the expected return value of the call corresponding to the \(k^{th}\)
linearization point. The following two lemmas prove the correctness using
induction on the linearization points and checking the different linearization
point cases separately.

\newcommand{\tr}{\mt{true}}
\newcommand{\fl}{\mt{false}}

\begin{lemma}\label{lem:sim}
  After \(k \geq 0\) linearization points, we have \(\lval{k} =
  \vvalk{k}\) except for \sref{lp:4} linearization points. For \(k \geq 1\), the
  \(\lret{k}\) values are false for \sref{lp:1},
  true for \sref{lp:2}, true for \sref{lp:3a} and false for \sref{lp:3b}.
\end{lemma}

\begin{proof}
  We prove the claim by induction on \(k\). For the base case of \(k=0\), the
  claim is true as \vval{} is initialized with the initial of the
  compare-and-swap object. Let \lpk{} be the \(k^{th}\) linearization point for
  \(k \geq 1\) and say that it corresponds to a call by a process \(i\). We have the
  following cases.

  \underline{\sref{lp:1}:} Let \lpkd{} be the linearization point previous to
  \lpk{}. By induction hypothesis, it holds that \(\lval{k'} = \vvalk{k'}\). By
  \sref{lem:vvalmodatlp}, the value of \vval{} does not change until \lpk{}. As
  we have a read operation at \lpk{}, it holds that \(\vvalk{k'} =
  \vvalk{k}\). By \sref{def:linp}, we know that \(\vvalk{k} \neq \aarg{i}\). So, it
  holds that \(\lval{k'} = \vvalk{k'} = \vvalk{k} \neq \aarg{i}\). Thus, it follows
  from the specification of the compare-and-swap object that \(\lval{k} =
  \lval{k'} = \vvalk{k}\). Moreover, we have \(\lret{k} = \fl\) as \(\lval{k'}
  = \vvalk{k} \neq \aarg{i}\).

  \underline{\sref{lp:2}:} Again, we let \lpkd{} to be the linearization point
  previous to \lpk{}. As argued in the previous case, it holds that \(\vvalk{k'} =
  \vvalk{k}\). By \sref{def:linp}, we know that \(\vvalk{k} = \aarg{i} = \barg{i}\). So,
  it holds that \(\lval{k'} = \vvalk{k'} = \vvalk{k} = \aarg{i}\). Thus, it follows
  from the object's specification that \(\lval{k} = \barg{i} =
  \vvalk{k}\). Further, we have \(\lret{k} = \tr\) as \(\lval{k'} = \aarg{i}\).

  \underline{\sref{lp:3a}:} Consider the point \lpkd{} when the value
  \(\vseqk{k} - 2\) was written to \vseq{} for the first time. As \(\vseqk{k}\)
  is even by \sref{lem:veven}, it follows from \sref{lem:vevenisalinp} that
  \lpkd{} is a linearization point or the initialization point.  Using
  definition of \sref{lp:3a}, \lpk{} is the first point when the value
  \(\vseqk{k}\) was written to the field \vseq{}. So, we have
  \(\vseqr{i} = \vseqk{k'}\). Thus, it holds that \(\vvalrdval{i} = \vvalk{k'}\)
  by \sref{lem:vseqeqthenvvaleq}.  Therefore, \(\vvalrdval{i} = \lval{k'}\) as
  \(\lval{k'} = \vvalk{k'}\) by induction hypothesis.  Using definition of
  \sref{lp:3a}, it also holds that \(\aarg{i} = \vvalrdval{i}\).  Thus, we have
  \(\aarg{i} = \lval{k'}\) and \(\lval{k} = \barg{i}\).

  Now, assume that the instruction at \lpk{} was executed by a process
  \(j\). Using definition of \sref{lp:3a}, we have \(i = \pidprd{j}\). As \lpk{}
  is the first time when the value of \vseq{} is \(\vseqk{k} = \vseqr{i} + 2\),
  we conclude that the process \(i\) is not finished until \lpk{} by using
  \sref{lem:vendcall}.  As \(\pseqr{j} = \vseqk{k} = \vseqr{i} + 2\), it is true
  that some process \(i'\) has \(\vseqr{i'} = \vseqr{i}\) and that the process
  executed \sref{ln:comp} until \lpk{}. As \(i = \pidprd{j}\), the process
  \(i' = i\). Moreover, the process \(i\) did this during the call corresponding
  to the linearization point \lpk{} as it follows from \sref{lem:vendcall} that
  there is a unique call for any process \(h\) given a fixed value of
  \(\vseqr{h}\). Thus, the process \(i\) already executed \sref{ln:ann} with
  \(\barg{i}\) as the value of the second field. This field has not changed as
  the call by process \(i\) is not finished until \lpk{}. So, we have
  \(\valard{j} = \barg{i}\) and that \(\vvalk{k} = \barg{i}\) as well. Because
  \(\aarg{i} = \lval{k'}\) as shown before, we also have \(\lret{k} = \tr\).

  \newcommand{\lpkdd}{\(\mt{LP}_{k''}\)}
  \newcommand{\vvalr}[1]{V.\mt{val}_{\ref*{ln:rdval}}^{#1}}
  \underline{\sref{lp:3b}:} Let \lpkd{} and \lpkdd{} be the first points when
  the value \(\vseqk{k}\) and \(\vseqk{k} - 2\) is written to \vseq{}
  respectively (\lpkd{} is just before the point \lpk{} as defined by
  \sref{lp:3b}). Let \(i\) and \(j\) be the processes that execute the calls
  corresponding to the points \lpk{} and \lpkd{} respectively. By definition of
  \sref{lp:3b}, we have \(\vseqr{i} = \vseqk{k'} - 2\).  As process \(j\) wrote
  \(\vseqk{k'}\) to \vseq{}, we have \(\vseqr{j} = \vseqk{k'} - 2\) as well. So,
  we have \(\vvalr{i} = \vvalr{j}\) using \sref{lem:vseqeqthenvvaleq}. Using
  definition of \sref{lp:3a} and \sref{lp:3b}, respectively, we have
  \(\aarg{j} = \vvalr{j} \neq \barg{j}\) and \(\aarg{i} = \vvalr{i}\). So, we
  have \(\aarg{i} \neq \barg{j}\).  We have \(\barg{j} = \lval{k'}\) as argued
  in the previous case, so it holds that \(\lval{k} = \lval{k'}\). By induction
  hypothesis, we have \(\lval{k'} =\vvalk{k'}\). Moreover, there no operations
  after \lpkd{} and until \lpk{} by definition of \sref{lp:3b}. So, we have
  \(\vvalk{k'} = \vvalk{k}\) and thus \(\lval{k} = \vvalk{k}\). Also, we have
  \(\lret{k} = \fl\) as \(\aarg{i} \neq \barg{j} = \lval{k'}\).
\end{proof}

\newcommand{\pp}{\,|\,}
\begin{lemma}\label{lem:ret}
  If the \(k^{th}\) linearization point for \(k\geq 1\) corresponds to a
  finished call
  by a process \(i\), then the value returned by the call is \(\lret{k}\).
\end{lemma}
\newcommand{\succret}[1]{\mt{success}_{\ref*{ln:ret}}^{#1}}
\begin{proof}
  Say the \(k^{th}\) linearization point is a \sref{lp:1} point. Using its
  definition, the value returned by the corresponding call is \(\fl\) as the
  condition in \sref{ln:neq} holds true. Using \sref{lem:sim}, we have
  \(\lret{k} = \fl\) as well for \sref{lp:1}. Next, assume that the \(k^{th}\)
  linearization point is a \sref{lp:2} point. Then, the value returned by the
  corresponding call is \(\tr\) as the condition in \sref{ln:eqeq} is true by
  definition. Using \sref{lem:sim}, we have \(\lret{k} = \tr\) as well for
  \sref{lp:2}.

  Now, consider that the \(k^{th}\) linearization point is a \sref{lp:3a}
  point. Say that the process \(j\) executes the operation at the linearization
  point. As \(\pidprd{j} = i\) by definition of \sref{lp:3a}, the process \(i\)
  already executed \sref{ln:comp} with the first field as \(\vseqk{k} - 1\).
  So, the process \(i\) also initialized \(R[i]\) to \((\cprd{j}\pp \fl)\) in
  \sref{ln:init}.  Moreover, the process \(j\) wrote the value
  \((\cprd{j}|\tr)\) to \(R[i]\) afterwards using a max-write operation.  Thus,
  the value of \(R[i].\mt{ret}\) after \lpk{} is \(\tr\). This field is not
  changed by \(i\) until it returns. And, other processes only write \(\tr\) to
  the field. So, the call returns \(\tr\) which is same as the value of
  \(\lret{k}\) given by \sref{lem:sim}.

  Next, consider that the \(k^{th}\) linearization point is a \sref{lp:3b}
  point. Let \(p\) be the point when the process \(i\) initializes \(R[i]\) to a
  value \((x\pp \fl)\) during the call (\sref{ln:init}). Consider a process \(j\) that tries to write \(\tr\) to
  \(R[i].\mt{ret}\) after \(p\) (by executing \sref{ln:inf}). So, it holds that
  \(\pidprd{j} = i\) and that \(\seqprd{j}\) is even. Now, we consider three
  cases depending on the relation between \(\seqprd{j}\) and \(\vseqk{k}\).
  First, consider that \(\seqprd{j} > \vseqk{k}\). As \(\pidprd{j} = i\) and
  \(\seqprd{j}\) is even, we have \(\vseqr{i} = \seqprd{j} - 2\) using
  \sref{lem:rlookup}. So, we have \(\vseqr{i} > \vseqk{k} -2\). This cannot
  happen until \(i\) finishes as \(\vseqr{i} = \vseqk{k} - 2\)
  for the current call by \(i\) using definition of \sref{lp:3b}.
  Second, consider that \(\seqprd{j} = \vseqk{k}\). Using definition of
  \sref{lp:3b}, there is a process \(h\) so that \(\pidprd{h} \neq i\) and
  \(\seqprd{h} = \vseqk{k}\). As \(\seqprd{j} =
  \vseqk{k}\) by assumption, we have \(\pidprd{j} \neq i\) using
  \sref{lem:ppideqifseqeq}. This contradicts our assumption that \(\pidprd{j} =
  i\).
  Third, consider that \(\seqprd{j} < \vseqk{k}\). As \(\pidprd{j} = i\) and
  \(\seqprd{j}\) is even, we have \(\vseqr{i} = \seqprd{j} - 2\) using
  \sref{lem:rlookup}. So, we have \(\vseqr{i} < \vseqk{k} -2\). This corresponds
  to a previous call by the process \(i\) as \(\vseqr{i} = \vseqk{k} - 2\) for
  the current call by \(i\). So, it holds that \(\card{j} < x\) and execution of
  \sref{ln:inf} has no effect. 
  Thus, the process \(i\) returns \(\fl\) for \sref{lp:3b}  which matches
  the \(\lret{k}\) value given by \sref{lem:sim}.

  If the \(k^{th}\) linearization point is a \sref{lp:4} point, then we know
  from \sref{lem:vendcall} that the call is unfinished and we need not consider
  it.
\end{proof}

We can now state the following main theorem about 
\sref{alg:sim}.

\begin{theorem}\label{thm:sim}
  \sref{alg:sim} is a wait-free and linearizable implementation of the compare-and-swap
  register where both the compare-and-swap and read functions take \(O(1)\) time.
\end{theorem}
\begin{proof}
  We conclude that the compare-and-swap function as given by
  \sref{alg:sim} is linearizable by using \sref{lem:dur} and
  \sref{lem:ret}. The read operation is linearized at the point of execution of
  \sref{ln:rd0}. Clearly, this is within the duration of the call. To check the
  return value, let \lpk{} be the linearization point of the read operation and
  \lpkd{} be the linearization point previous to \lpk{}. Then, we have
  \(\vvalk{k} = \vvalk{k'}\) using \sref{lem:vvalmodatlp}. So, it holds that
  \(\vvalk{k} = \lval{k'}\) using \sref{lem:sim}. Moreover, both the
  compare-and-swap and read functions end after executing \(O(1)\) steps and the
  implementation is wait-free.
\end{proof}

\section{Consensus Numbers}\label{sec:cn}
In this section, we prove that each of the max-write and the half-max primitives has consensus
number one. Note that these are two separate claims. One, that it is impossible to solve consensus
for two processes using read-write registers and registers that support the max-write and read
operation. Second, that it is impossible to solve consensus for two processes using read-write
registers and registers that support the half-max and read operation.  Trivially, both operations
can solve binary consensus for a single process (itself) by just deciding on the input value. To
show that these operations cannot solve consensus for more than one process, we use an
indistinguishability argument.

First, we define some terms. A \emph{configuration} of the system is the value of
the local variables of each process and the value of the shared registers. The
\emph{initial} configuration is the input \(0\) or \(1\) for each process and
the initial values of the shared registers. A configuration is called a bivalent
configuration if there are two possible executions starting from the
configuration so that in one of them all the processes terminate and decide
\(0\) and in the other all the processes terminate and decide \(1\). A
configuration is called \emph{\(0\)-valent} if in all the possible executions
starting from the configuration, the processes terminate and decide
\(0\). Similarly, a configuration is called \emph{\(1\)-valent} if in all the
possible executions starting from the configuration, the processes terminate and
decide \(1\). A configuration is called a univalent configuration if it is
either \(0\)-valent or \(1\)-valent. A bivalent configuration is called
\emph{critical} if the next step by any process changes it to a univalent
configuration.
Consider an initial configuration
in which there is a process \(X\) with the input \(0\) and a process \(Y\) with
the input
\(1\). This configuration is bivalent as \(X\) outputs \(0\) if it is made to run until it
terminates and \(Y\) outputs \(1\) if it is made to run until it terminates.
As the terminating configuration is univalent, a critical configuration is
reached assuming that the processes solve wait-free binary consensus.

Assume that the max-write operation can solve consensus between two processes
\(A\) and \(B\). Then, a critical configuration \(C\) is
reached. W.l.o.g., say that the next step \(s_a\) by the process \(A\) leads
to a \(0\)-valent configuration \(C_0\) and that the next step \(s_b\) by the
process \(B\) leads to a \(1\)-valent configuration \(C_1\). In a simple
notation, \(C_0 = C s_a\) and \(C_1 = C s_b\). We have the following cases.
\begin{enumerate}
\item \(s_a\) and \(s_b\) are operations on different registers: The
  configuration \(C_0 s_b\) is indistinguishable from the configuration \(C_1
  s_a\). Thus, the process \(B\) decides the same value if it runs until
  termination from the configurations \(C_0 s_b\) and \(C_1 s_a\), a
  contradiction. 
\item \(s_a\) and \(s_b\) are operations on the same register and at least one
  of them is a read operation: W.l.o.g., assume that \(s_a\) is a read operation. Then,
  the configuration \(C_0 s_b\) is indistinguishable to \(C_1\) with respect to
  \(B\) as the read operation by \(A\) only changes its local state. Thus, the
  process \(B\) decides the same value if it runs until termination from the
  configurations \(C_0 s_b\) and \(C_1\), a contradiction.
\item \(s_a\) and \(s_b\) are write operations on the same register: Then, the
  configuration \(C_0 s_b\) is indistinguishable from the configuration \(C_1\)
  as \(s_b\) overwrites the value written by \(s_a\). Thus, the process \(B\) will decide
  the same value if it runs until termination from the configurations \(C_0 s_b\) and \(C_1\), a contradiction.
\item \(s_a\) and \(s_b\) are max-write operations on the same register \(R\):
  Say that the arguments of these operations are \(a \p x\) and \(b \p y\) for
  \(A\) and \(B\) respectively. W.l.o.g., assume that \(b \geq a\). Then, there are following two
  cases.
  \begin{enumerate}
  \item Operation \(s_b\) does not modify the register \(R\). Thus, operation \(s_a\) will also
    leave it unchanged as \(b \geq a\). Also, the contents of \(R\) in \(C_1 s_a\) is same as in
    \(C_0\) because \(s_b\) did not modify \(R\) by assumption. So, the configuration \(C_1 s_a\) is
    indistinguishable from the configuration \(C_0\) with respect to \(A\) and it will decide same
    value if run until termination from the two configurations, a contradiction.
  \item Operation \(s_b\) modifies the register \(R\). In this case,    the
  configurations \(C_0 s_b\) is indistinguishable from \(C_1\) as \(b \geq a\) and the operation
  \(s_b\) will overwrite both the fields of the register \(R\). Thus, the process
  \(B\) will decide the same value from these configurations, a contradiction.
  \end{enumerate}
\end{enumerate}
So, the critical configuration cannot be reached and the processes \(A\) and \(B\)
cannot solve consensus using the max-write primitive. Thus, its consensus number is one.

For the half-max primitive, we do a similar case analysis. The first three cases are the same as in
the case of max-write primitive. For
the last case, assume that \(s_a\) and \(s_b\) are half-max operations on the same register
\(R\). Say that the argument of these operations are \(a\) and \(b\) for processes \(A\) and \(B\)
respectively. Assume w.l.o.g. that \(b \geq a\). We have the following two cases as before.
\begin{enumerate}
  \item Say that \(s_b\) does not modify \(R\). In this case, even \(s_a\) does not modify \(R\) as
    \(b \geq a\). Thus, the contents of \(R\) is  same in the configurations \(C_0\) and \(C_1
    s_a\) and these configurations are indistinguishable to \(A\). So, it will decide the same value
    if run until termination from these configurations, a contradiction.
  \item Say that \(s_b\) modifies the register \(R\). In this case, the first half of the register
    \(R\) in  the configurations \(C_0 s_b\) is same as the first half of \(R\) in \(C_1\). This is
    because \(s_b\) overwrites the first half of \(R\) in both the configurations \(C_0\) and
    \(C\). The second half is not modified by either \(s_a\) or \(s_b\) so the contents of \(R\) is
    same in \(C_0 s_b\) and \(C_1\). Therefore, these configurations are indistinguishable with
    respect to \(B\) and it will decide the same value if run until termination from these
    configurations, a contradiction.
\end{enumerate}
So, the critical configuration cannot be reached and the processes \(A\) and \(B\)
cannot solve consensus using the half-max primitive. Thus, its consensus number is one as well.

\section{Conclusion}\label{sec:conc}
The algorithm that we presented simulates  a single compare-and-swap register using \(O(n)\) registers
that support the half-max, max-write, read and write primitives. If \(m\) compare-and-swap registers are to be
simulated, then a straightforward approach requires \(O(mn)\) registers. However, we can improve
this if we observe that there is at most one pending operation per process even if \(m\)
compare-and-swap registers
have to be simulated. The arrays \(A\) and \(R\) store the information about the latest pending
call per process so there is no need to allocate them  for every compare-and-swap register. Only the
registers \(P\) and \(V\) need to be allocated separately. As the counter value \(c\) used in the
first half of each entry of array \(A\) or \(R\) is always increasing, we will be conceptually
running \(m\) instances of the presented algorithm using \(O(m+n)\) registers. Actually, if one observes
closely, the three fields used in the register \(P\) are useful only when more than one
compare-and-swap registers need to be implemented. Otherwise, we can use a single counter replacing
both \(c\) and \(\mt{seq}\).

One issue with the presented algorithm is that it uses unbounded sequence numbers. Thus, the algorithm only
works if the size of the registers is at least logarithmic in the total number of compare-and-swap
operations executed. Actually, the growth in sequence numbers can be much slower as out of the two
unbounded counter types, one of them counts the total number of compare-and-swap operations executed
per process and the other one counts the total number of successful compare-and-swap operations
only. Also, as a first step towards understanding the power of a set of weak instructions with respect to
their ability to efficiently simulate compare-and-swap, we did not
focus on bounding the  sequence numbers.

Using our result, one can transform any \(O(T)\) time algorithm that uses compare-and-swap and
read-write registers into an \(O(T)\) time algorithm that uses reasonably large registers and
support the instructions half-max, max-write, read and write.  As the
transformation is wait-free, it even works for algorithms that are not wait-free. But, is it also
true that any \(O(T)\) time algorithm using registers that support half-max, max-write, read and
write instructions can be transformed into an \(O(T)\) time algorithm using compare-and-swap and
read-write registers? There is an \(\Omega(\log n)\) lower bound \cite{jayanti:timeComplexityLower}
on information aggregation among \(n\) processes which applies to compare-and-swap and read-write
registers but not to registers that support half-max, max-write, read and write instructions. Thus,
it may be possible that there are tasks that take \(o(\log n)\) time using max-write, half-max, read
and write registers but \(\Omega(\log n)\) time using compare-and-swap and read-write registers.

There are other practical factors too that can affect efficiency. For example, the half-max and the
max-write operations are associative and do not return a value. Thus, they can be easier to combine
in the processor memory interconnect when there is contention for the same memory location.  In this
paper however, we only show that a set of weak instructions can be theoretically at least as good as
compare-and-swap with respect to time complexity. Although this highlights the power of a set of
weak instructions, it also opens up the question that what is the best set of synchronization
instructions in general.

\bibliography{cas}

\end{document}